\documentclass[12pt,a4paper,reqno
]{amsart}
\usepackage{graphicx}
\usepackage{amsmath,amssymb,amsthm,array}
\usepackage{amsfonts, amsmath, amsthm, amssymb}
\newcommand{\sign}{\mathop{\mathrm{sign}}\nolimits}%

\newtheorem{remark}{Remark}
\newtheorem{assumption}{Assumption}
\newtheorem{theorem}{Theorem}
\newtheorem{corollary}{Corollary}
\newtheorem{example}{Example}

     \title[ Superposition of shock waves of BBM equation] {Superposition of shock waves of the generalized BBM equation}%
 \author{Alexey Samokhin}\vspace{6pt}
\address{Institute of Control Sciences of Russian
Academy of Sciences
65 Profsoyuznaya street, Moscow 117997, Russia}\vspace{6pt}
\email{ samohinalexey@gmail.com}\vspace{6pt}
\begin{document}
\maketitle

\begin{abstract}
 The generalized BBM
 studied in this paper contains an additional dissipative term.
 Thus instead of solitons for the classic BBM there exists a lot of travelling shock wave solutions. The rules of  their interactions or superposition is of high importance. The paper gives a detailed description of the two-parameter family of travelling wave solutions and proves their stability using a conservation law. Based on these results, effective rules of superposition are obtained. Moreover these rules are applicable not exclusively to the travelling wave solutions of BBM, but also to a  wider class of shock waves,  in particular discontinuous.
Characteristic examples are illustrated by numerically worked out  graphs.
\end{abstract}



\maketitle

\section*{Introduction}

The  Benjamin-Bona-Mahony equation (BBM)  is of the form

\begin{equation}\label{0}
u_t + u_x + uu_x - \lambda u_{xxt}  = 0.
\end{equation}
It is
also known as the regularized long-wave equation.
 It first  appeared (with $\lambda=1$) in  \cite{1},                       as a modification of the Korteweg-de Vries equation (KdV).
 long wavelength in liquids, acoustic-gravity waves in compressible
fluids, hydromagnetics waves in cold plasma and acoustic waves in anharmonic crystals, etc.

 The generalized BBM (gBBM) studied in this paper  has one  additional
  dissipation term:
 \begin{equation}\label{1}
u_t + u_x + uu_x - \lambda u_{xxt} - \varepsilon u_{xx} = 0.
\end{equation}

Here $\lambda>0$, $\varepsilon>0$ are constants connected to dispersion
and dissipation.

      Instead of solitary waves characteristic for classical BBM \eqref{0}, the combination of dissipation and dispersion produce continuous monotonic or oscillatory traveling waves. Throughout this paper they are called travelling shock waves (as well as "real" shocks, i.e., with discontinuities). They are the subject of this paper.

There is a considerable number of publication in the field.  The following citations are the most relevant.
  In \cite{2}
the interaction of the BBM \eqref{0} solitary waves is  described using conservation laws; it also contains a vast bibliography.  In \cite{3} a superposition of the KdV-B travelling waves is discussed.
 The review \cite{4} deals with modelling Rieman-type shocks for
$u_t=f(u)_x$ type equations  using conservation laws,
with modified KdV (mKdV) and KdV-B as a main examples. In \cite{5} travelling wave solutions and conservation laws for the KdV-BBM equations are explored.

This paper content is as follows.

Section \ref{types} contains the detailed description of the gBBM
continuous travelling waves  in terms of their boundary conditions.

In section \ref{stable} the stability of
travelling continuous  shocks is proven
using the selective decay rate (see e.g, \cite{6}), of the BBM \emph{momentum} conservation law.

   In section \ref{super}  simple and effective rules are obtained for
superposition of  travelling shocks, both  continuous and discontinuous.

 In section \ref{ex}  few characteristic examples are illustrated by
  computer simulated graphs.

The section \ref{discus} states the main results and contain some ideas for further research.

\section{Travelling wave shocks solutions\label{types}}

The travelling wave solutions of \eqref{1} are Galilean invariant

\begin{equation}\label{2}
u = u(x - V t).
\end{equation}

The profile of such solution stay unaltered while travelling along the $x$-axis with velocity $V$.

Substitution of \eqref{2} into \eqref{1} results in ODE
  \begin{equation}\label{3}
(-V + 1)u_x + uu_x + \lambda V u_{xxx} -\varepsilon u_{xx} = 0.
\end{equation}

The sole known explicit formula for TWS was obtained via the
 solution of \eqref{3} and is of the form

  \begin{eqnarray}\label{4}
V -1
-\frac{3\varepsilon^2 }{25\lambda V}\left[\tanh^2\left(\frac{\varepsilon (x+s-V t)}{10\lambda V}  \right)+2\tanh\left(\frac{\varepsilon v}{10\lambda V}\right)- 1\right]
\end{eqnarray}

 For the chosen medium $\{\varepsilon, \lambda\}$ it describes a one-parameter family denoted below as
$F_{V,s}$ where $V \in \mathbb{R}\setminus {0}$; the $s$ is an arbitrary translation along $x$-axis. Therefor  $s$ is often omitted
when it does not play any role in an argument.

 The expression \eqref{4} could be obtained by twice lowering the order of \eqref{3}.
First, integration of \eqref{3} leads to

  \begin{equation}\label{5}
(-V + 1)u +\frac{u^2}{2}+\lambda V u_{xx}-\varepsilon u_x=C
\end{equation}

Second, put $u'=p(u)$, then $u_{xx}= u''=pp'$.

 It follows

\begin{equation*}\label{6}
(-V + 1)u +\lambda pp'-\varepsilon p=C
\end{equation*}

The latter equation is a particular case of the second kind of the
Abel equation $qq'- q = g(z)$. In this case one explicit solution of it may be found for a
specific $g(z)$. Returning this solution to the initial variables gives \eqref{4}, see e.g. \cite{3}.

\begin{remark} The $V = 0$ (i.e., the stationary solutions) are described by  the equation

\[u_x + uu_x - \varepsilon u_{xx} = 0.
\]

Its general solution include

  \begin{equation}\label{8}
u(x) = -2C \tanh\left(
\frac{C}{\varepsilon}(x + s)\right)-1,
\end{equation}

with $u(x)|_{x=\pm \infty} = -1\mp 2C.$
\end{remark}

\begin{remark}
Any constant $u(x,t)=C$ is also a traveling wave solution,  for arbitrary $V$.
\end{remark}

 Instead of $V$ and $C$ in \eqref{5} , it is convenient to chose
 the boundary conditions  $u(x)|_{x=\pm \infty}$ as the other two parameters.

 The main object of study in this paper is the travelling
 shock waves. To be precise, the following assumption is taken below.

\begin{assumption}\label{ass}

For a travelling shock wave
 \begin{eqnarray}
\nonumber(\exists\; u(x)|_{x=-\infty}=H \; &\wedge &\; \exists\; u(x)|_{x=+\infty}=h)\\
\bigwedge (\forall n>0: \; \frac{d^nu(x)}{dx^n}|_{x=\pm \infty}=0).&&
  \end{eqnarray}
 \end{assumption}

Here one can see a clear analogy between the latter conditions and the
 requirement
 \[
\forall n\in\mathbb{N}: \; \frac{d^nu(x)}{dx^n}|_{x=\pm \infty}=0
\]

commonly used in mathematical  physics to describe a
localized phenomena (e.g. solitons); such functions are called "rapidly decreasing at infinity".

Note that $ u_x(x)$ for a travelling shock rapidly decreases
 at infinity.

Applying the above assumption to \eqref{5} we get

  \begin{eqnarray}
    (-V+1)H+\frac{H^2}{2} &=& C, \\
    (-V+1)h+\frac{^2}{2} &=& C.
  \end{eqnarray}
Substraction  leads to

\[
(H-h)\left(-V+1+\frac{H+h}{2}\right)=0.
\]

If $H=h$ then $V$ is arbitrary, then $u=const=H$. Otherwise

\begin{equation}\label{10}
  V=1+\frac{H+h}{2},\quad C=-\frac{H\cdot h}{2}.
\end{equation}

It shows that a family of shock TWS (denoted below $ T_{H,h}$) is  2-dimensional. Also the formula \eqref{10} gives an explicit connection between the coefficients of \eqref{5} and boundary conditions.

In particular, for the explicit TWS \eqref{4}
\[ T_{H,h}=
F_V=T_{V-1\pm\frac{6\varepsilon^2}{25\lambda V}, V-1\mp \frac{6\varepsilon^2}{25\lambda V}}
\]

To describe the structure of this $\{(H,h)\}$ family consider, by a standard scheme,  the dynamical system associated with \eqref{5}:

\begin{equation*}\label{11}
  \begin{array}{rcl}
    u'&= & p \\
    \lambda V p'&=& (V-1)u -\frac{u^2}{2}-
\varepsilon p+C.
  \end{array}
\end{equation*}

The fixed points satisfy  conditions $u'=0,; p'=0$, so
$(V-1)u -\frac{u^2}{2}+C=0$. When using \eqref{10} it comes to

\[2(H+h)u-u^2- Hh=0\;\Rightarrow\; u_1=H,\;u_2=h.
\]

Thus the fixed points are, naturally, $(H,0)$ and $(h,0)$. A solution of \eqref{5} is the separatrix connecting them on the corresponding directional field.

Now find the types of these fixed points. Solve the characteristic equation

\begin{equation*}\label{12}
  \det
  \begin{pmatrix}
           \frac{\partial u'}{\partial u}-k & \frac{\partial u'}{\partial p} \\
           \frac{\partial p'}{\partial u} & \frac{\partial p'}{\partial p}-k \\
         \end{pmatrix}=\det
  \begin{pmatrix}-k&1\\
  \frac{V-1-u}{\lambda V}& \frac{\varepsilon}{\lambda V}-k
         \end{pmatrix}=0,
\end{equation*}

that is

\begin{equation*}\label{13}
  k^2-\frac{\varepsilon}{\lambda V}k-\frac{(V-1-u)}{\lambda V}=0.
\end{equation*}

The roots are

  \begin{equation*}
  k_\pm =\frac{\varepsilon}{2\lambda V} \pm
  \sqrt{\left( \frac{\varepsilon}{2\lambda V}\right)^2+
\frac{(V-1-u)}{\lambda V}}
\end{equation*}

where  $u=H$ or $u=h$.

So, by \eqref{10}, $V-1-u=\frac{H+h}{2}-u=\pm\frac{H-h}{2}$.

The real and complex value roots are separated by  the sign of

  \begin{equation}\label{14}
\left( \frac{\varepsilon}{2\lambda V}\right)^2+
\frac{(V-1-u)}{\lambda V}=
 \frac{1}{2\lambda V^2}
\left( \frac{\varepsilon^2}{2\lambda}\pm \left(1+\frac{H+h}{2}\right) (H-h)\right)
\end{equation}

(the choice between $+$ and $-$ here depends on the choice of a fixed point).

For
\[
\left|\left(1+\frac{H+h}{2}\right) (H-h)\right|
\]

big  enough one of a fixed points has complex roots  $k_\pm$; that point  correspods to a focus on the phase portrait of the dynamical system. Another pointthen has $k_\pm$ real a and of different signs; a saddle.

 Since $\lambda V^2>0$, the zeros of
 \eqref{14} defines the  hyperbola  on the $(H,h)$ plane:

\[
 \left(1+\frac{H+h}{2}\right) (H-h)=\pm\frac{\varepsilon^2}{2\lambda}.
 \]

Hence the criterion "monotonic versus oscillatory" is

  \begin{equation}\label{criterion}
  \sign\left[\left|(H+1)^2-(h+1)^2\right|- \frac{\varepsilon^2}{\lambda}\right].
 \end{equation}

The hyperbola asymptotes are

\begin{itemize}
\item $\left(1+\frac{H+h}{2}\right) (H-h)=0$ (its points correspond to  $V=0$, i.e. to solutions \eqref{8},
\item $H=h$ (its points correspond to constants)
solutions).
 \end{itemize}

Recall that for explicit solution \eqref{4},

\[
 H,h=V-1 \pm\frac{6\varepsilon^2}{25 \lambda V},
   \]

 so the sign of $V$ coincides with the sign of $H-h$ and

  \[\left(1+\frac{H+h}{2}\right) (H-h)=+\frac{12}{25}
\]

[

\begin{figure}[h]
\begin{minipage}{14pc}
\includegraphics[width=14pc]{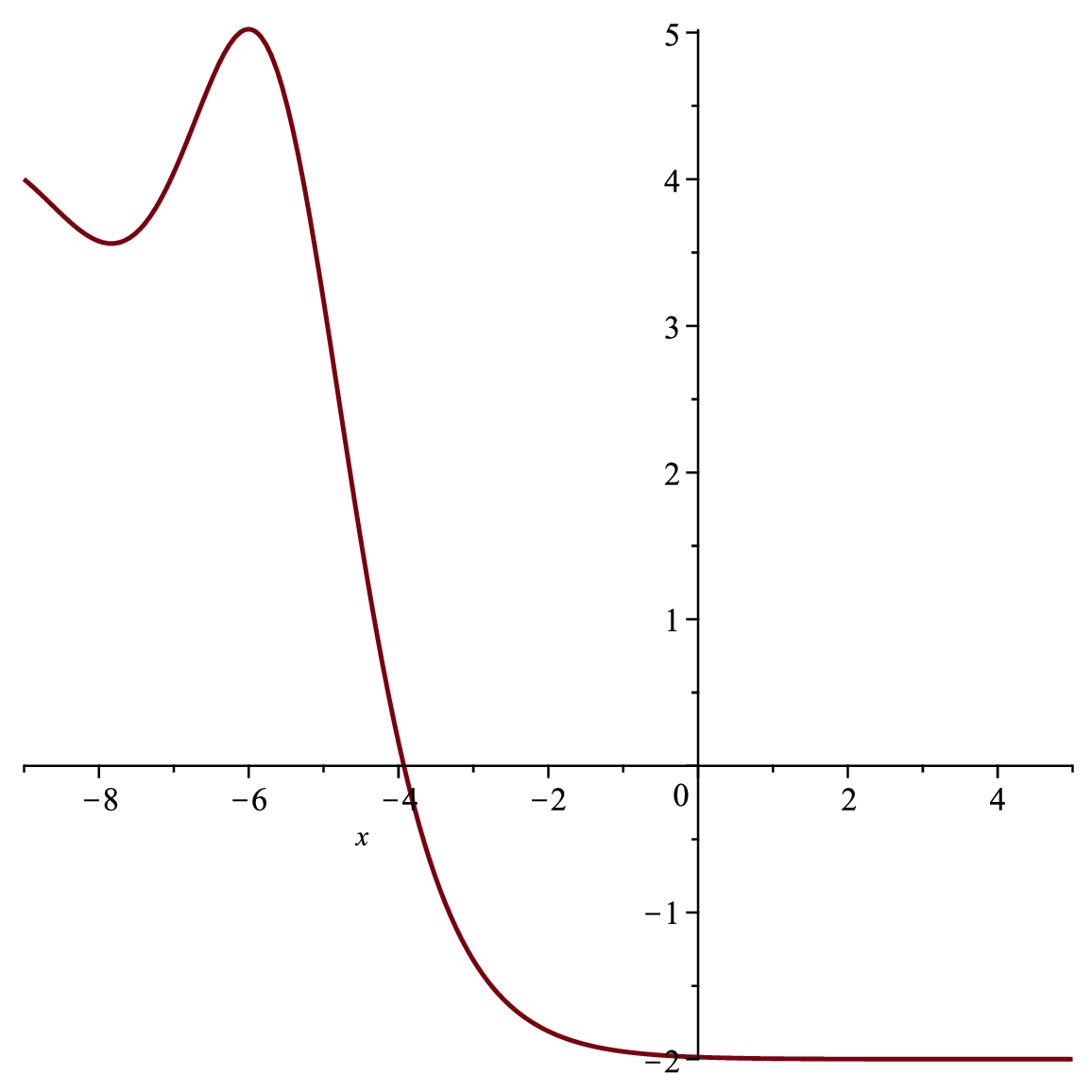}
\end{minipage}\hspace{2pc}%
\begin{minipage}{14pc}
\includegraphics[width=14pc]{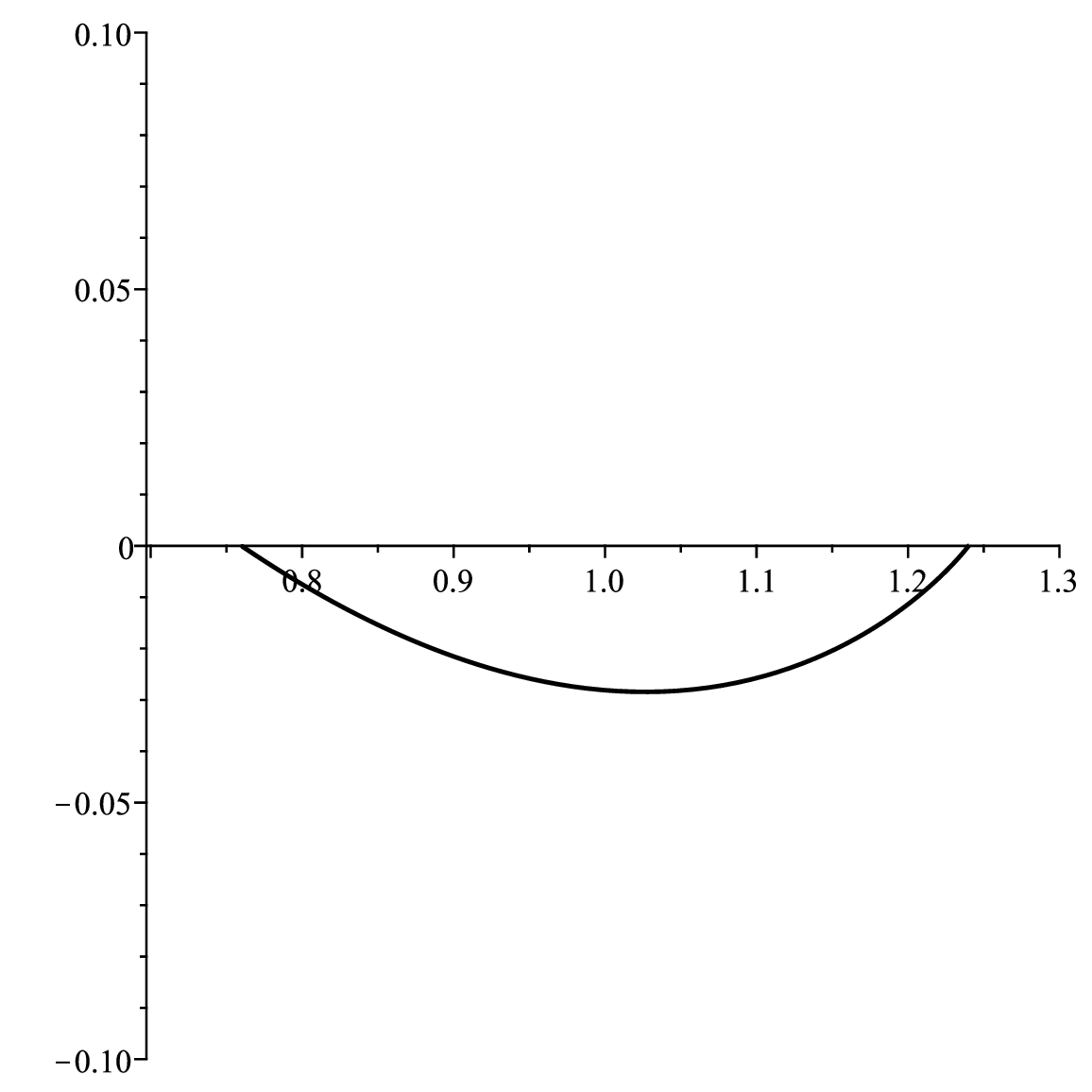}
\end{minipage}
\caption{\small \textbf{Left:} Monotonic solution \eqref{4}. \textbf{Right:} Its phase portrait.
$(\lambda=0.5,\;\varepsilon=1,\; H=1.24,\;h=0.76)$
\label{TWSphase}
}
\end{figure}

\begin{figure}[h]
 \begin{minipage}{14pc}
{\includegraphics[width=14pc]{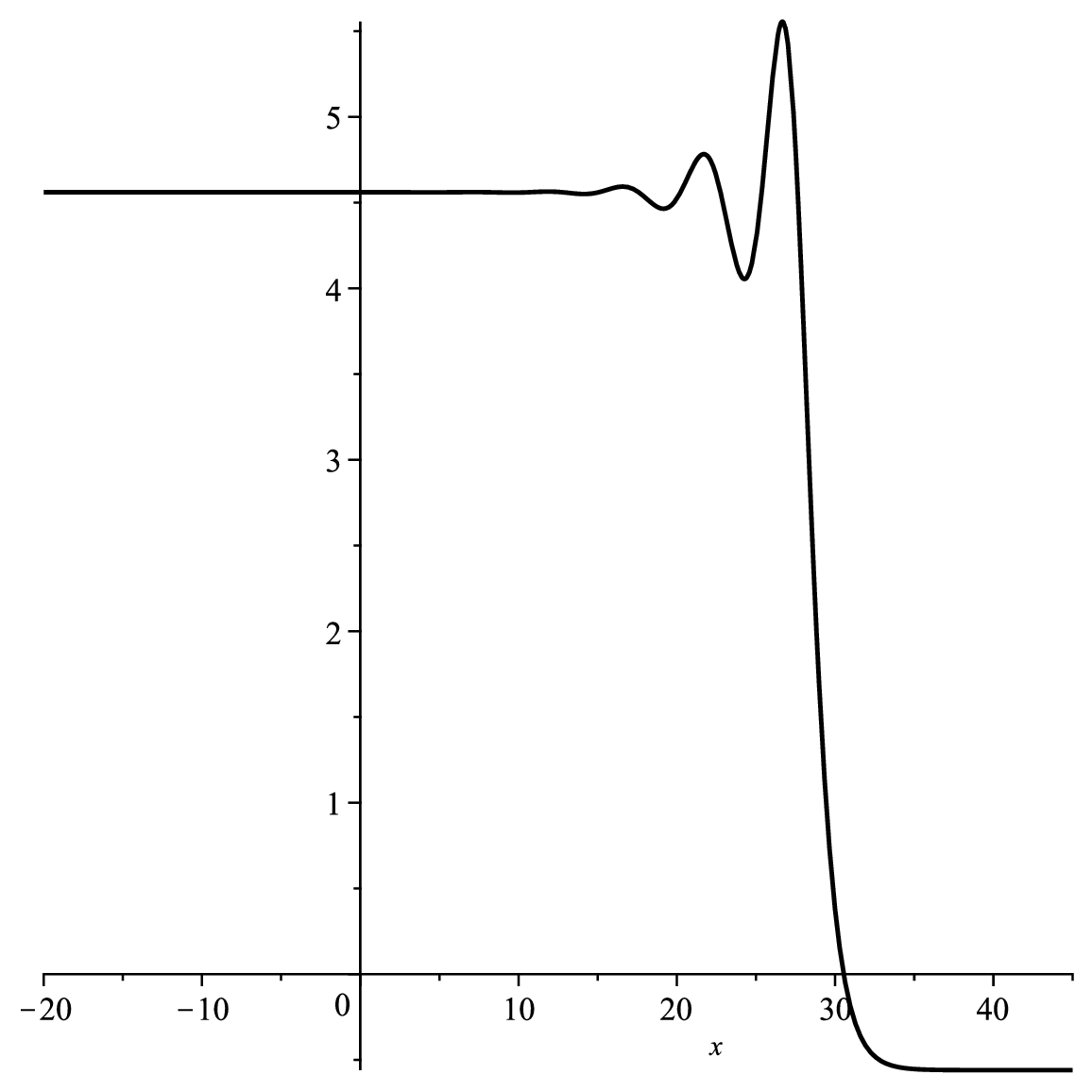}}
\end{minipage}
 \begin{minipage}{14pc}
{\includegraphics[width=14pc]{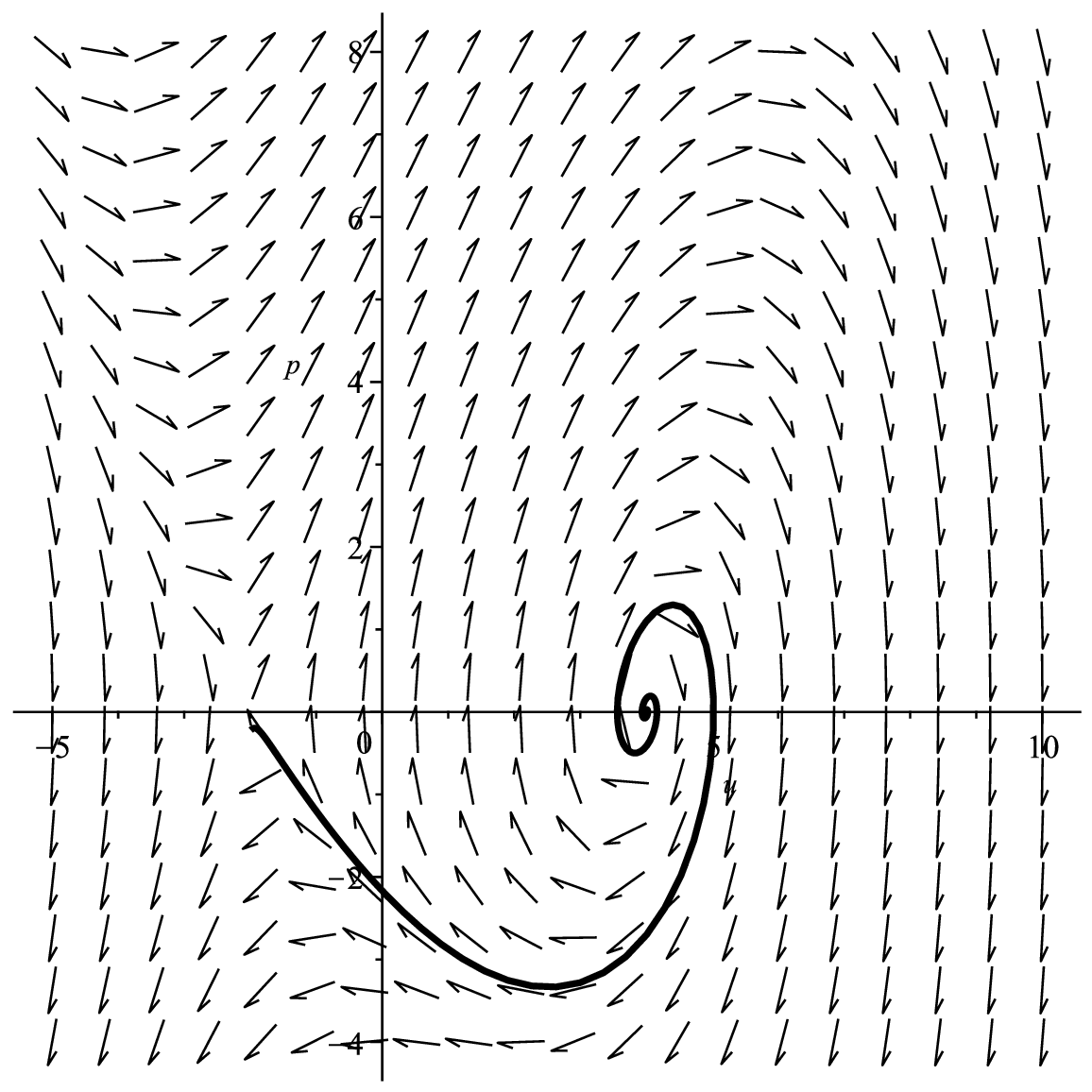}}
\end{minipage}
\caption{ \small \textbf{Left:} Oscillatory solution:  \textbf{Right:} Its phase portrait.
$(\lambda=0.5,\;\varepsilon=1,\; H=4.24,\;h=-2)$
\label{focus}
}
\end{figure}

The following diagram shows the distribution of types of TWS solutions on the $(H,h)$ plane (c.f. \cite{3} with analogous 
diagram for mKdV).

\begin{figure}[h]
\includegraphics[width=0.85\textwidth]{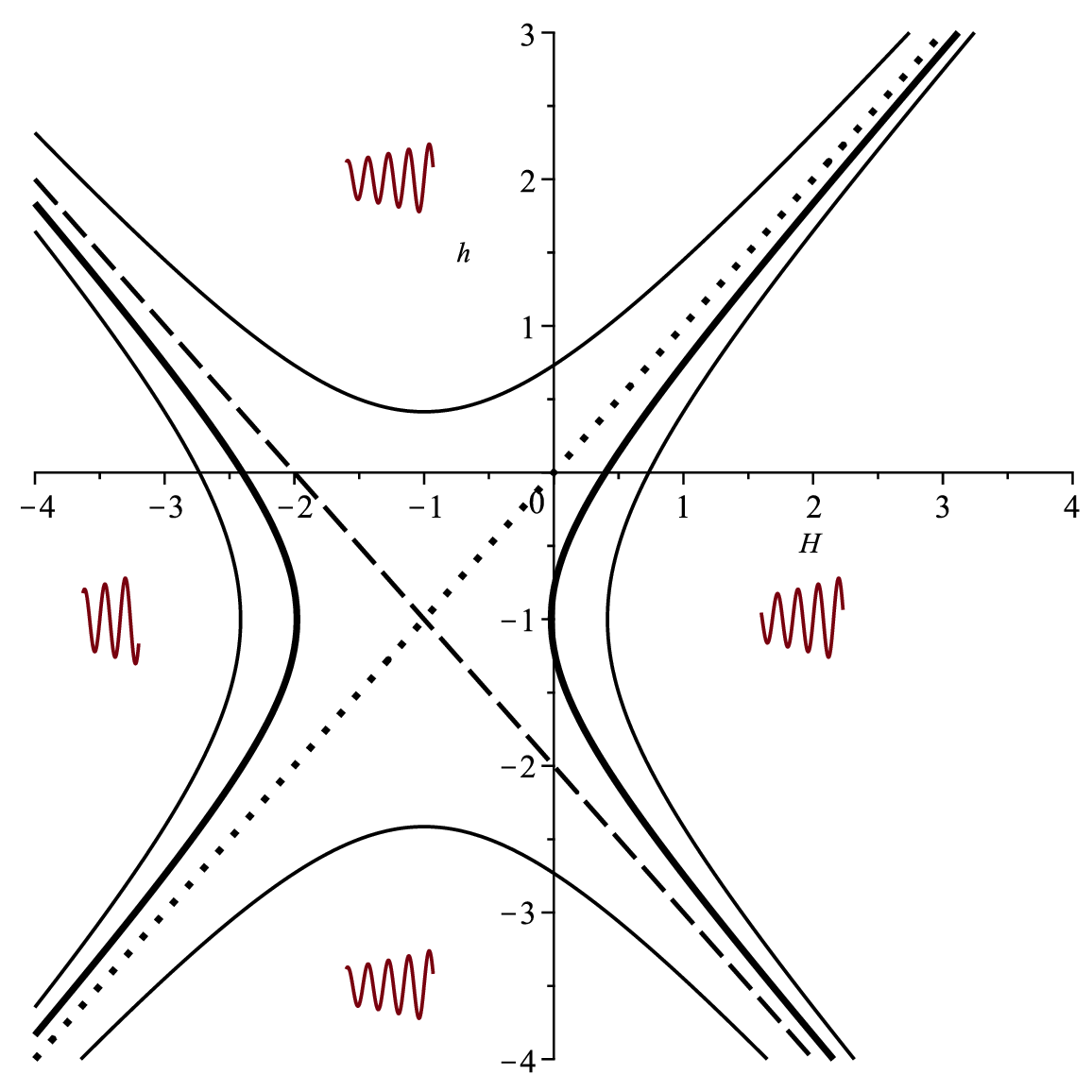}
\caption{\small
Types of TWS, medium:  $\{\varepsilon=1, \lambda=0.5\}$.
 \label{diag}}
\end{figure}

 The diagram contains:

\begin{itemize}
\item Internal Regions within hyperbolas (marked by a pictogram)--- one complex root, focus/oscillations;
\item Region between  hyperbolas --- real roots/monotonic;
\item Dash asymptote $H=h$  --- constant TWS;
 \item Dots asymptote $V=1+\frac{H+h}{2}=0$  --- stationary TWS, of the form \eqref{8};
 \item Thick solid line --- explicit TWS, \eqref{4} type.
 \end{itemize}

 \subsection{The range  of oscillations}

  For oscillatory TWS  it is important to asses
   the range of oscillation.

 As it can be seen on Figure \ref{focus} (and examples below), the
 growth of oscillation stops at a point where the
     changing level $u(x,t)$ of
 the wave surface surpass the boundary of $[H,h]$. So $|H-h|$ is an upper estimate of the range, admittedly
  rough.

 For a more precise estimate recall that near the focus $H$, up to translation along $x$-axis,
\[ u(x) \approx \Re(H+\exp(k x)),\;k\in\mathbb{C},
  \]
  or
\[
 u(x) \approx H+\exp(\Re(k) x)\sin(\Im(k) x).
\]

Consider the case $H>h$.  Since  approximate extremal values of oscillations are

$ \pm \exp({\Re(k)}x),$

the inequality  $H+\pm \exp(\Re( k)x)>h$ holds. It remains to find a the first extremal left to the first solution of $\exp({\Re(k)}x)>h-H$.

Recall that $\Re(k)=\frac{\varepsilon}{2\lambda V}$.

Also
 the distances between extremum is
 $\pi / \Im(k)$, it is not hard to find a domain of
 suitably significant oscillations.

 \section{Stability of shock waves\label{stable}}
                                               n
        There is a lot of solutions of \eqref{1} with the boundary condition $H,h$: just solve  \eqref{1}
        for the initialvalue/boundary problem

  \[          \{u(x,0)=v(x), u(-\infty,t)=H, u(+\infty,t)=h\},
  \]
   where $v(x)$ is not a solution of \eqref{5}.

 For the next argument we put the
w      weak restrictions on $v(x)$:

\[
\left\{
  \begin{array}{ll}

(\exists\; v(x)|_{-\infty}=H \; \wedge \; \exists\; v(x)|_{+\infty}=h)\bigwedge (\forall n>0: \; \frac{d^nu(x)}{dx^n}|_{\pm \infty}=0).
  \\[3mm]

\exists\int_\mathbb{R}v_x^2(x)\,dx<\infty.
  \end{array}
\right.
\]
 Denote by $S_{H,h}$ the class of equation \eqref{1}
 solutions obtained with such a  $v(x)$  as the initial condition. Note that  $T_{H,h}\in S_{H,h}$.

\begin{theorem}
[The  stability property of travelling shock $T_{H,h}:$]

Every solution from $S_{H,h}$  tends to $T_{H,h}$ as $t\rightarrow\infty$.
\end{theorem}

\begin{proof}

As it is well known, e.g.  \cite{2}, the equation BBM \eqref{0} has  the \emph{momentum}  conservation law, namely  $u^2$. That means that for all rapidly decreasing at infinity solutions $u(x,t)$ of  \eqref{0}

\[
\frac{\partial}{\partial t}\int\limits_\mathbb{R}u^2\;dx=
2\int\limits_\mathbb{R}uu_t\;dx=
-2\int\limits_\mathbb{R}u(u_x + uu_x - \lambda u_{xxt}  )\;dx=0
\]

holds. 

Not so for the equation gBBM \eqref{1}. There the \emph{momentum}  conservation law $u^2$ decays:

\begin{equation*}\label{decay}
\begin{array}{ll}
\frac{\partial}{\partial t} \int\limits_\mathbb{R}u^2dx=
-2\left[\int\limits_\mathbb{R}u(u_x + uu_x - \lambda u_{xxt})\;dx+
\varepsilon\int\limits_\mathbb{R}uu_{xx} dx\right];\\[3mm]
\mbox{so    } \frac{\partial}{\partial t} \int\limits_\mathbb{R}u^2dx=2\varepsilon\int\limits_\mathbb{R}uu_{xx} dx=-2\varepsilon\int\limits_\mathbb{R}u_x^2 dx<0.
\end{array}
\end{equation*}

Thus a positive \emph{momentum} decays (or dissipates) monotonically to zero while $u_x\neq 0$. And, inevitably, $u\rightarrow 0$ --- to its boundary condition.

 Let $T_{H,h}(x,t)$ be a travelling shock wave with $H,h$ as the boundary conditions, and a perturbation $v(x)$  of $T_{H,h}(x,0)$ with the same boundary conditions.
 Solving
 \eqref{1} for initial value/boundary problem $\{\delta|_{x=\pm\infty}=0;
\delta(x,0)=T_{H,h}(x,0)-v(x)\}$
results in the solution $\delta(x,t)$ rapidly decreasing at infinity.

As it was stated above, the \emph{momentum} of $\delta$ decays and

\[
\lim\limits_{t\rightarrow+\infty}\delta(x,t)=0.
\]

 It means that the initial difference between $T_{H,h}(x,0)$ and $v(x)$  vanishes.
\end{proof}

\section{Superposition of shocks } \label{super}

There are two immediate corollaries from the stability theorem.

\begin{corollary} [Superposition rule.]

\begin{equation}
T_{h_1,h_1}+T_{h_1,h_1}\rightarrow T_{H_1+H_2,h_1+h_2}
\end{equation}
\end{corollary}
\begin{proof}

Consider two shock waves: $T_{H_1,h_1}$
 and $T_{h_2,h_2}$.
 The sum of these expressions is not, of course, a solution to the  nonlinear equation gBBM \eqref{1}.
Yet

\[
T_{h_1,h_1}+T_{h_1,h_1}=
\mbox{ has the boundary values same as}
T_{H_1+H_2,h_1+h_2}.
\]

Now just take $\delta(x,0)=[T_{H_2,h_2}+T_{H_1,h_1}-T_{H_1+H_2,h_1+h_2}]|_{t=0}$
in the above
\emph{momentum} argument.
.
As it follows from the Theorem,

\begin{equation}\label{vectors}
T_{h_1,h_1}+T_{h_1,h_1}\rightarrow T_{H_1+H_2,h_1+h_2}
\end{equation}
\end{proof}

In particular
\[
F_{V_1}+F_{V_2}\rightarrow T_{H_1+H_2,h_1+h_2},
\]

where

\[
H_i=V_i-1+\frac{6\varepsilon^2}{25\lambda V_i},\; h_i=V_i-1\frac{6\varepsilon^2}{25\lambda V_i},
\]
and the type of resulting shock wave is defined by the criterion{criterion}.

\begin{remark}
 The equation \eqref{vectors} means that a superposition of two
 TWS corresponds to  sum of vectors
 $\overrightarrow{(H_1,h_1)}$  and  $\overrightarrow{(H_2,h_2)}$
  on the plane (see Figure \ref{diag} diagram).
  Then elementary geometry or algebra show that the
  superposition of two explicit type \eqref{4}  waves
 $F_{V_1}$ and $F_{V_2}$   is \emph{never} of the same explicit type.
\end{remark}

\begin{corollary}[Superposition velocity]

\[ V=V_1+v_2-1\]
\end{corollary}
\begin{proof}

 Nonlinear superposition $T(H,h)=T_{H_1+H_2,h_1+h_2}$  has the velocity $V$
      \[
V=1 +\frac{H_1+H_2+h_1+h_2}{2}=1 +\frac{H_1+h_1}{2}+1 +\frac{H_2+h_2}{2}-1=V_1+V_2-1.
\]
Thus
 \[
 V=V_1+v_2-1
 \]
\end{proof}

\begin{remark}
It is not that only solutions of (2) may be chosen as summands. Almost any kind of rapidly decreasing at infinity perturbations
will suit. Even shock waves with discontinuities eventually tends to TW
shocks: boundary conditions rule!
\end{remark}

\section{Examples \label{ex}}

Below few computer simulated characteristic examples are given.  They illustrate the process of
 a wave sum transformation into one TWS.

 On  all graphs in this section the thin solid
 line stands for the initial sum, the thick solid line ---
 for the resulting  final TWS;
 dash lines show intermediate states.

\begin{example}[Turbulent union, Figure~\ref{turbulent}.]

The medium in this example is $\{\varepsilon=1,\lambda=0.5\}$.

Initial $F_{V_1}$ and $F_{V_2}$  summands from medium $\{\varepsilon=1, \lambda=0.1\}$ have
$V_1=2.5, s_1=6$ and $V_2=1.5, s_2=0$ .
\end{example}

For the initial sum $(H,h)=(4.56,-0.56)$.

\begin{figure}[h]
\includegraphics[width=0.75\textwidth]{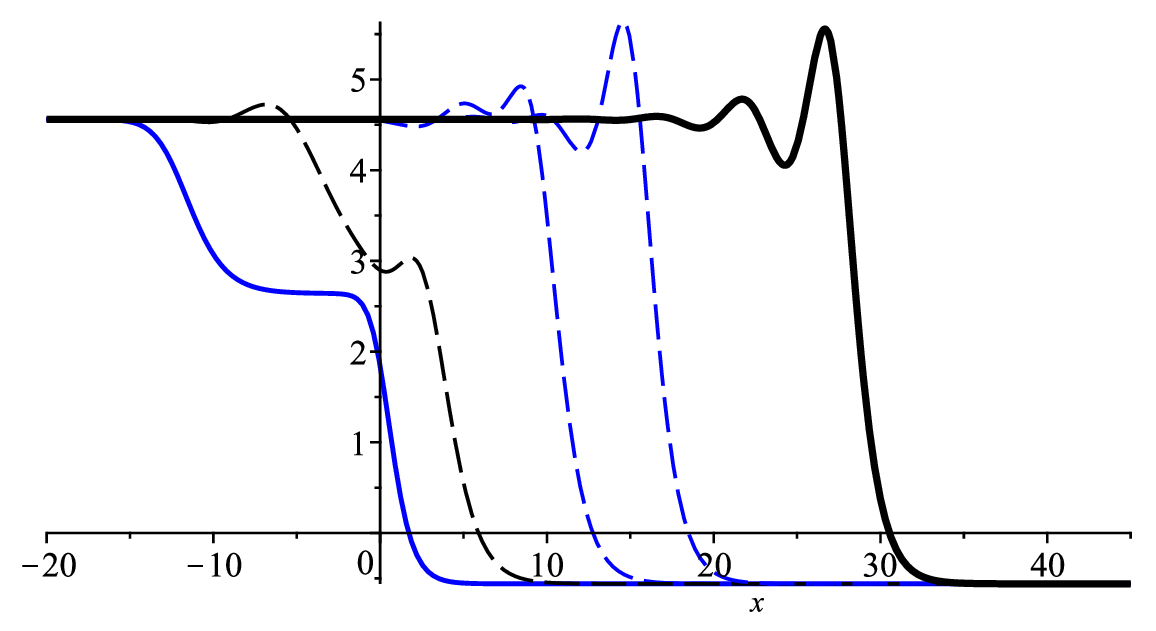}
\caption{\small
Race results in turbulent union $T(4.56,-0.56)$,  $V=3$.
 \label{turbulent}
 }
 \end{figure}

 Initial distance between the waves fronts is $s=6$, the
approach speed is $V_1-V_2=1$. So after $t=6$ the fronts join at $x=9$. This common front, after reformatting, proceeds in an oscillation form $T(H,h)$ at speed $V=3$.

To asses the oscillation rate use $k=13+\frac{i}{3}\sqrt{14.36}$.

It is important to note that the summands
 are \emph{ not} solutions in the example's chosen  medium,  since their exact formulas are obtained from a different medium. That is, in this modelling summands are just general form shock waves.

\begin{example}[Turbulent union of shocks, Figure \ref{shock}.]

Another case with non-solutions superposition:

The medium in this example is $\{\varepsilon=1,\lambda=0.5\}$.
The Figure~\ref{shock} presents the evolution of the sum
of two Rieman breaks:

 \[f(x,t)=\left\{
            \begin{array}{ll}
              0.54, & x<Wt+6; \\
              2.46, & x>Wt+6.
            \end{array}
          \right.,\\,
g(x,t)=\left\{
            \begin{array}{ll}
              2.1, & x<Vt; \\
              -1.1, & x>Vt.
            \end{array}
        \right.,\;
\]

$W=2.5,\;V=1.5.$

\end{example}

\begin{figure}[h]
\includegraphics[width=0.75\textwidth]{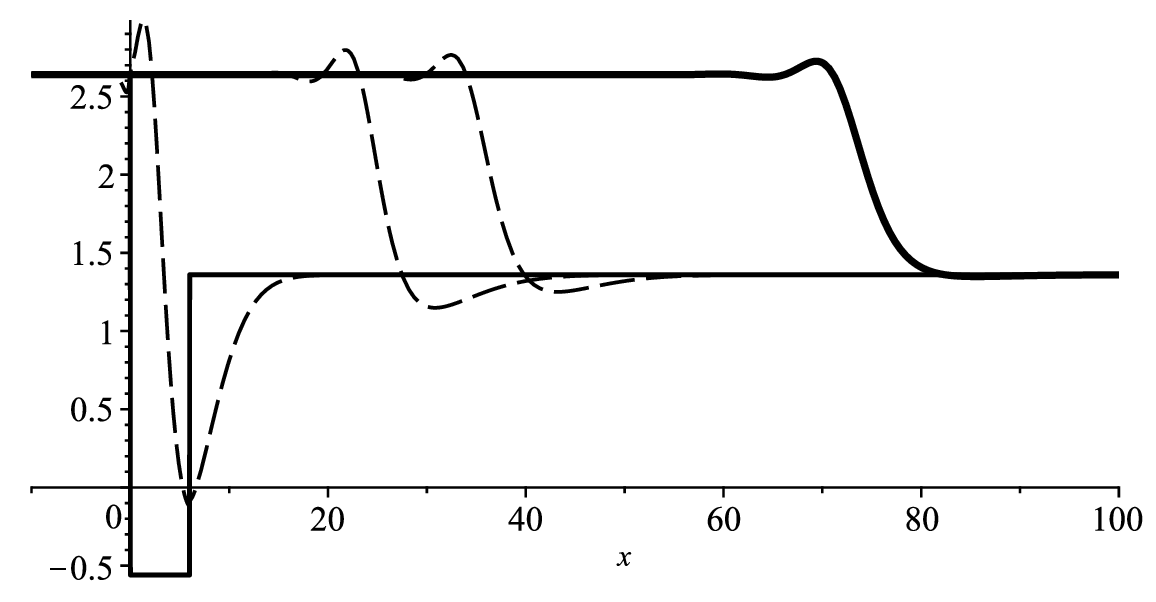}
\caption{\small
Discontinuity results in a turbulent union
$T(2.64,1.36)$,  $V=3$.
 \label{shock}
 }
 \end{figure}

\begin{example}[Peaceful fusion, Figure~\ref{fusion}]

The medium in this example is  $\{\varepsilon = 1, \lambda = 0.05\}$.

Initial $F_{V_1}$ and $F_{V_2}$ summands are not solutions:
they are obtained in a different
medium,  $ s_1 = 6, V_1 = 1.5$
  and $ s_2=0, V_2=2.5$.
 Thus for the initial sum $(H,h)=(4.56,-0.56)$.
\end{example}

\begin{figure}[h]
\includegraphics[width=0.75\textwidth]{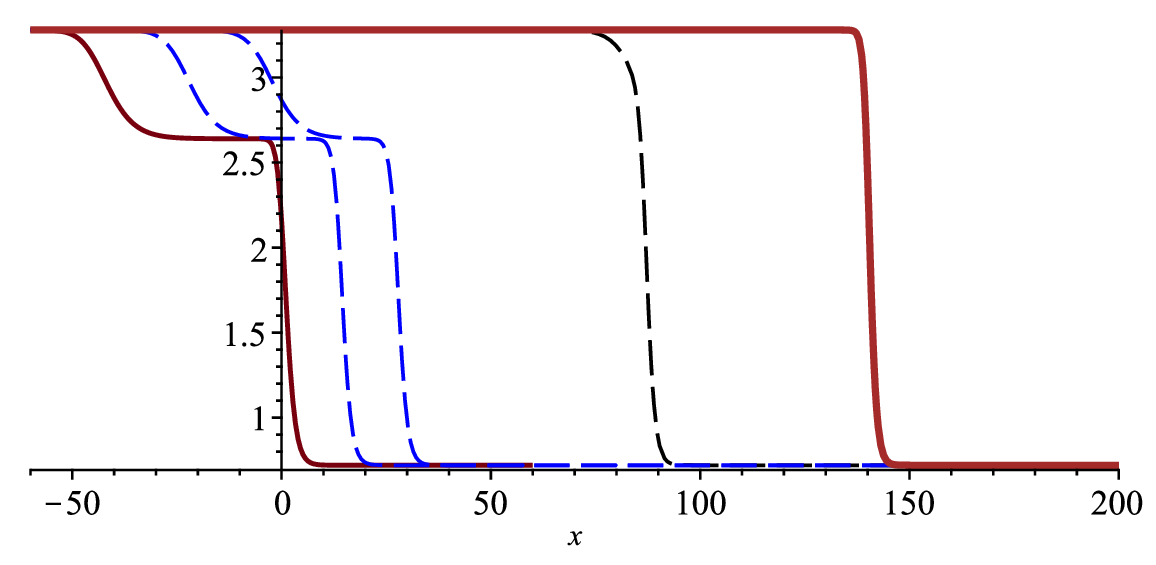}
\caption{\small
Race results in a peaceful fusion of a monotonic type,  $V=3$
\label{fusion}}
  \end{figure}
  After $t\approx 20$ the fronts of the summands join. The common front,
   after a reformatting, form their superposition and proceeds
   in an monotonic  form $T(H,h)$ at speed $V=3$.

But   the resulting monotonic TWS is not given by
 the formula \eqref{4}  since this example medium
$H=4.56\neq\frac{6\varepsilon^2}{25 \lambda V}+V-1$.

\begin{example}[Turbulent collision, Figure~\ref{collision}]

Medium:  $\{\varepsilon=1, \lambda=0.5\}$, common for both for the equation and summands.

Initial $F_{V_1}$ and $F_{V_2}$  summands: $ s_1 = 0, V_1 = 0.5$
and $ s_2 = -12, V_2 = -2.5$: waves are moving in opposite directions.

\end{example}

For the initial sum $(H,h)=(-2.8493 -5.1520)$

This time the summands are also the solutions in the chosen medium.

\begin{figure}[h]
\includegraphics[width=0.75\textwidth]{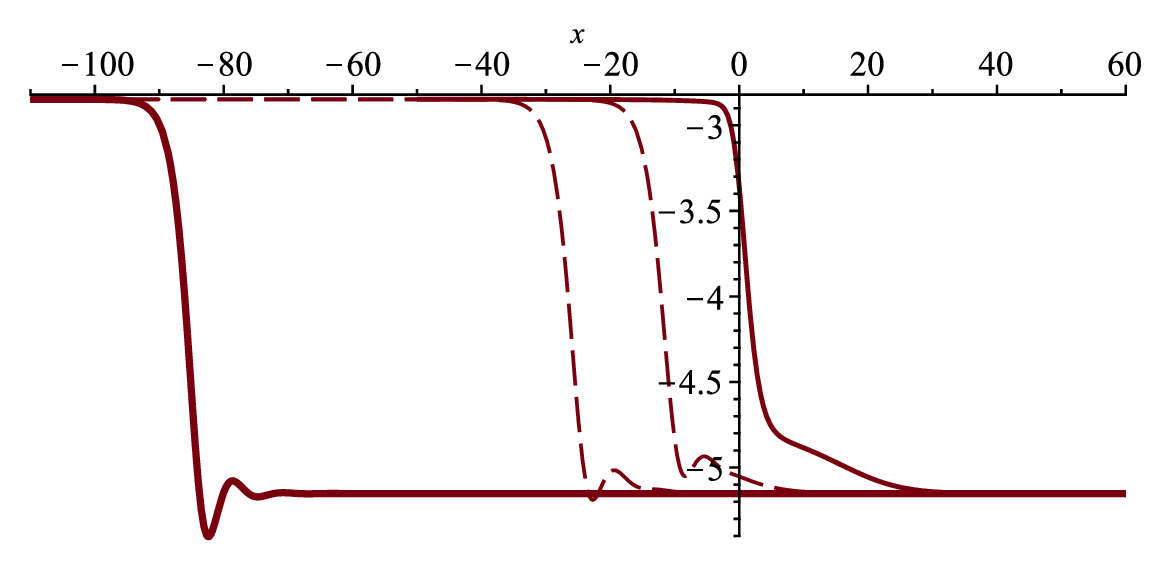}
\caption{\small
Collision results in turbulent union, $V=-3$.
\label{collision}
}
\end{figure}

After collision of two waves, the resulting TWS solution of a turbulent (oscillatory) type proceeds to the left.

Initial distance between the waves fronts is $s=12$, the
approach speed is $3$. So after $t\approx 4$ the fronts join. This common front, after reformatting proceeds in a form $T(H,h)$ at speed $V=-3$.

\begin{example} [Standstill, Figure~\ref{standstill}]

, Medium: $\varepsilon=1, \lambda=0.1$, common for both for the equation and summands.

Initial $F_{V_1}$ and $F_{V_2}$  summands $V_1=2, s_1=10$ and $V_2=-1, s_2=-10$.

Here $(H,h)_1=1\pm\frac{6}{5}$ and $(H,h)_2=-2\mp\frac{6}{2.5}$.
\end{example}

\begin{figure}[h]
\includegraphics[width=0.75\textwidth]{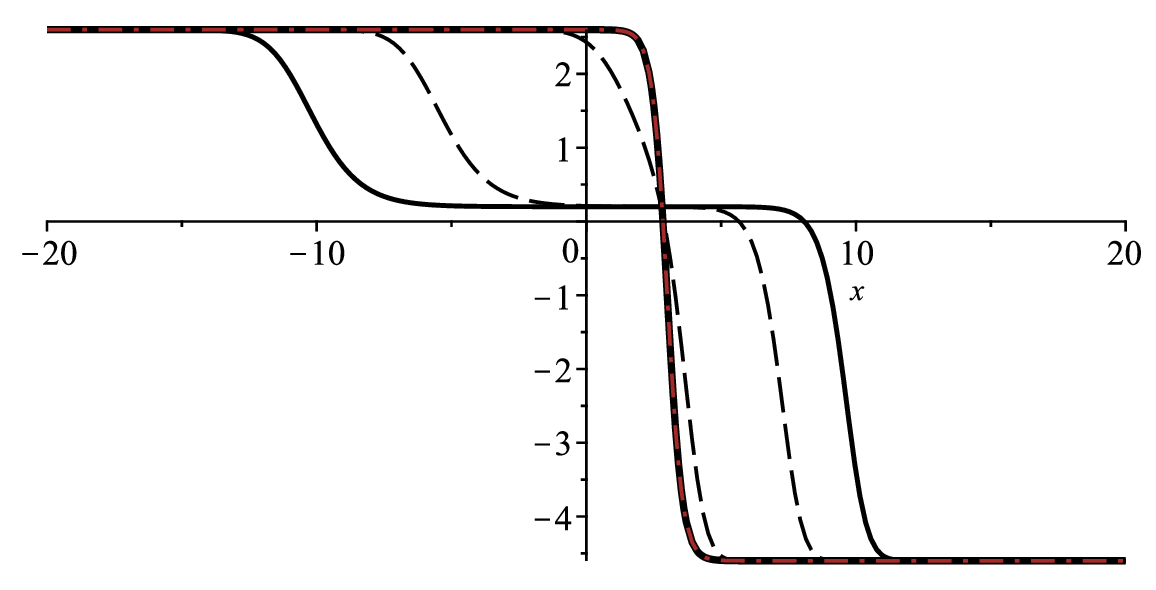}
\caption{\small
Collision results in standstill \eqref{8}, $V=0$.
\label{standstill}
 }
 \end{figure}

 Initial distance between the waves fronts is $s_1-_2=21$, the
approach speed is $V_1-V_2=3$. So after $t=7$ the fronts join at $x=3$. . This common front, after rotation around $(x,0)$, stops at the position in the form  \eqref{8}.

 The equation for the thick solid line --- for the resulting  stationary TWS:
\[
u=1-3.6 \tanh(1.8(x-3));
\]

\section{Discussion\label{discus}}
 The gBBM
  differs from the classical BBM by an additional dissipation term.

  As a result, instead of the BBM solitons  appears a class of travelling shock wave solutions  --- solutions with fixed values at infinity and all  derivatives rapidly decrease there. The way of  interaction  (superposition) of such solutions is of high
 practical importance.

  In this paper a detailed  description of the two-parameter family of the gBBM shock wave solutions is obtained  and their stability is
  proved using the \emph{momentum } conservation law. Based on these results, effective rules of superposition are obtained. Moreover these rules are applicable not exclusively to the travelling
   wave solutions of gBBM, but also to a much wider class of shock waves.

  This study fortify the  usefulness of the selective decay approach
   to study of diffusion/dissipanive equations.

  The characteristic examples are illustrated by numerically worked out  graphs;  Maple 21 standard packages: DEtools, PDEtools and Plots were used in this study.

\section*{Acknowledgments}
  \textbf{Funding:} This research received no external funding.

\end{document}